\documentclass[submission,copyright,creativecommons]{eptcs}
\providecommand{\event}{LSFA 2024} 

\usepackage{iftex}

\ifpdf
  \usepackage{underscore}         
  \usepackage[T1]{fontenc}        
\else
  \usepackage{breakurl}           
\fi

\usepackage{graphicx}
\usepackage{hyperref}
\usepackage{tikz}
\usepackage{amsthm}
\usepackage{amssymb,amsmath,mathtools,float,multicol}
\usepackage{nicefrac,xfrac}
\usepackage{xspace}
\usepackage{mathrsfs}
\usepackage{stmaryrd} 
\DeclareFontFamily{U}{mathb}{}
\DeclareFontShape{U}{mathb}{m}{n}{
  <-5.5> mathb5
  <5.5-6.5> mathb6
  <6.5-7.5> mathb7
  <7.5-8.5> mathb8
  <8.5-9.5> mathb9
  <9.5-11> mathb10
  <11-> mathb12
}{}
\DeclareSymbolFont{mathb}{U}{mathb}{m}{n}
\DeclareFontSubstitution{U}{mathb}{m}{n}
\DeclareMathSymbol{\sqdoublecap}{2}{mathb}{"5E}
\DeclareMathSymbol{\sqdoublecup}{2}{mathb}{"5F}
\usetikzlibrary{positioning}
\usetikzlibrary{fit}
\DeclareMathAlphabet{\mathcal}{OMS}{cmsy}{m}{n} 
\def \justif#1{\scriptstyle{\{\text{#1}\}} \\}
\def\twist#1{\mathcal{#1}}
\def\A{\mathcal{A}}
\def\T{\mathcal{T}}
\def\K{\mathcal{K}}

\DeclareMathOperator*{\bigsqdoublecup}{\scalebox{1.5}{$\sqdoublecup$}}
\newcommand{\Rel}{\mathbf{Rel_P}}
\newcommand{\Set}{\mathbf{Set_P}}
\newcommand{\model}{(W,R,V)}

\newcommand{\doublef}{f\mspace{-7mu}f}
\newcommand{\doublet}{t\mspace{-3mu}t}
\tikzset{
block/.style = {draw=white,rectangle,text width = 3em,align=left, rounded corners=0ex, minimum height=2.0em}
}
\definecolor{red1}{rgb}{139, 0.0, 0.0} 
\definecolor{blue1}{rgb}{0.0, 0.0, 139}
\definecolor{magenta1}{rgb}{0.55, 0.0, 0.55}

\newtheorem{theorem}{Theorem}
\newtheorem{definition}{Definition}
\newtheorem{lemma}{Lemma}
\newtheorem{example}{Example}
\newtheorem*{remark}{Remark}
\allowdisplaybreaks

\title{Paraconsistent Relations as a Variant of Kleene Algebras \thanks{This work was financed by PRR - Plano de Recuperação e Resiliência under the Next Generation EU from the European Union within Project Agenda ILLIANCE C644919832-00000035 - Project n 46, as well as by National Funds through FCT - Fundação para a Ciência e a Tecnologia, I.P. (Portuguese Foundation for Science and Technology) within the project IBEX, with reference PTDC/CCI-COM/4280/2021 (DOI 10.54499/PTDC/CCI-COM/4280/2021)}}
\author{
Juliana Cunha 
\institute{CIDMA, Dep. Mathematics, Aveiro University \\ Aveiro, Portugal}
\institute{INESC TEC \& Dep. Informatics, Minho University \\ Braga, Portugal }
\email{juliana.cunha@ua.pt}
\and
Alexandre Madeira
\institute{CIDMA, Dep. Mathematics, Aveiro University \\ Aveiro, Portugal}
\email{\quad madeira@ua.pt}
\and
Luís S. Barbosa
\institute{INESC TEC \& Dep. Informatics, Minho University \\ Braga, Portugal }
\email{lsb@di.uminho.pt}
}
\def\titlerunning{Paraconsistent Relations as a Variant of Kleene Algebras}
\def\authorrunning{J. Cunha, A. Madeira \& L. S. Barbosa}
\begin{document}
\maketitle

\begin{abstract}
Kleene algebras (\texttt{KA}) and Kleene algebras with tests (\texttt{KAT}) provide an algebraic framework to capture the behavior of conventional programming constructs.
This paper explores a broader understanding of these structures, in order to enable the expression of programs and tests  yielding vague or inconsistent outcomes. Within this context, we introduce the concept of a paraconsistent Kleene Algebra with tests (\texttt{PKAT}), capable of capturing vague and contradictory computations. Finally, to establish the semantics of such a structure, we introduce two algebras, $\Set(\T)$ and $\Rel(\K,\T)$, parametric on a class of twisted structures $\K$ and $\T$. We believe this sort of structures, for their huge flexibility, have an interesting application potential.
\end{abstract}

\section{Introduction}
In his seminal work \cite{Kle56}, Stephen Kleene described finite deterministic automata along with a specification language: regular expressions. Kleene's paper left open the question of whether a finite, sound, and complete axiomatization of the equivalence of regular expressions existed, which would provide an algebraic framework for describing regular languages.
This question has been explored by many researchers. In 1964, Redko \cite{Red64} proved that no finite set of equational axioms could fully characterize the algebra of regular expressions, and in 1966, Salomaa \cite{Sal96} provided two complete axiomatizations of this algebra. Conway's comprehensive 1971 work \cite{Con71} presented a detailed overview of results related to regular expressions and their axiomatizations. In 
\cite{Kozen94} Kozen showed that Salomaa’s axiomatization is non-algebraic,i.e., unsound under substitution of alphabet symbols by arbitrary regular expressions. He then presented an algebraic axiomatization: Kleene algebras (\texttt{KA}).
Later in \cite{Kozen97}, Kozen also introduced a variant of \texttt{KA} called Kleene Algebras with Tests (\texttt{KAT}). The addition of tests to \texttt{KA} was prompted by the need to express conventional constructs such as conditionals and \textbf{while} loops. As a result, \texttt{KAT} is specifically designed for equational reasoning about these constructs. For any proposition $\alpha$ it is possible to form a test $\alpha?$ that acts as a guard: it succeeds with no side effects in states satisfying $\alpha$, and fails or aborts otherwise\footnote{Notice that throughout this paper,  distinct symbols will be used for programs and propositions. Therefore, we often omit the $?$ symbol and simply write $\alpha$ to denote a test.}.

While historically prominent in automata theory and formal languages \cite{SalomaaKuich}, \texttt{KA} and their variants have found applications across various domains, including relational algebra \cite{Tarski41}, program semantics and logics \cite{Pratt1990}, compiler optimization \cite{KP2000}, and algorithm design and analysis \cite{kozen1992design}. For a more comprehensive read of its applications, the reader is referred  to \cite{backhouse}.

 \texttt{KAT} are suitable to reason about imperative programs since these programs can be thought of as  sequences of discrete steps, each related to an atomic transition in a standard automaton. Traditionally, these programs operate within a bivalent truth space, where assertions have Boolean outcomes. While this framework has proven to suit a huge range of computer science applications, its inherent simplicity and rigidity may fall short in capturing some intricacies present in real-world scenarios.

Actually, it is not uncommon in software engineering to face scenarios in which vague, or weakly consistent, or even contradictory information is present. Often such characteristics cannot be abstracted away or swept under the carpet. A typical example the authors are currently facing, and which forms one of our motivations for this work, concerns repositories of medical images in a particular domain, which are labeled by several medical judgments from different expert teams, which often are in partial contradiction. Another example arises in the analysis of implementations of quantum circuits in current NISQ (Noisy Intermediate-Scale Quantum) technology \cite{Preskill18}, where conflicting decoherence levels in the quantum memory have to be taken into account.

Therefore, the introduction of algebraic structures to model computations becomes necessary when the behavior of these computations does not conform to a simple bivalent outcome. Instead, it may involve weighted outcomes from a richer domain, potentially lacking consistency. In this context, vagueness captures the lack of information, while paraconsistency, a well-established designation in logic, expresses the excess of information arising from contradictory judgments.

From a technical point of view, paraconsistency refers to a property inherent to a consequence relation. A logic is said to be ``paraconsistent'' if and only if its logical consequence relation, whether semantic or proof-theoretic, does not lead to explosion \cite{Priest07}. In logic, the term ``explosion'' refers to the principle of \emph{ex falso quodlibet}, meaning ``from contradiction, anything follows''. This principle is the basis for the law of non-contradiction in classical logic. It asserts that from contradictory premises, any proposition can be derived, thus leading to triviality (where anything can be proved true). Consequently, paraconsistent logics set themselves apart from classical ones by their capacity to handle inconsistent information without ``exploding'' into absurdity. Initially developed in Latin America during the 1950s and 1960s, notably through the influential contributions of F. Asenjo and Newton da Costa, paraconsistent logic quickly garnered attention within the logic and computer science communities. Its original focus on mathematical applications has since expanded, as evidenced by recent literature emphasizing the engineering potential of paraconsistency \cite{Aka16}. Relevant applications are documented in several fields,  including deontic logic \cite{Costa1986}, data network monitoring \cite{Cortes22}, robotics \cite{emmyiii}, quantum mechanics \cite{CG00} and quantum information theory \cite{AC10},

Since their introduction numerous formalizations of paraconsistent logic have emerged \cite{DaCosta1977,Dunn2002}. This paper takes a specific standpoint: the notion of paraconsistent transition systems (PLTS) introduced in \cite{lsfa22} and later used to reason about decoherence in quantum circuits as documented in \cite{BarbosaM23}. 
Informally, in these systems, each transition is assigned a pair of weights: a \emph{positive} weight representing the evidence supporting the transition's occurrence, and a \emph{negative} weight representing the evidence against it.

This paper aims at developing an algebraic counterpart  to our research on PLTS, already documented in a number of references  \cite{BarbosaM23,NCL22,lsfa22,TASE23,FSEN23}. As in previous works, we adopt a similar approach to that in \cite{Bou}, focusing on a particular class of residuated lattices over a set $A$ of possible truth values. In this setting, both the positive and negative weights are taken from the set $A$.

Furthermore, to establish the groundwork for the sequel, we work with the concept of a \emph{twisted-structure}, originally proposed by Kalman \cite{Kalman1958}. This structure arises from the direct product of a lattice $\mathbf{L}$ with its order-dual $\mathbf{L}^\partial$ and serves as a key tool for jointly computing positive and negative weights in PLTS, represented by pairs in $L \times L^\partial$. Additionally, the twisted-structure naturally carries a De Morgan involution, which we denote by $\sslash$ and interpret as a form of ``negation''.

The main contribution of this paper lies in the proposal of  an extension of \texttt{KAT} to a paraconsistent framework  able to reason about uncertain or inconsistent computations. This allows computations to yield outcomes graded by two weights: one indicating  evidence for execution and the other evidence for failure. Note that a similar motivation can be found in \cite{GMB19} where a graded variant of \texttt{KAT} is introduced. That work, however, only captures forms of uncertainty as usual in fuzzy logic. Finally, we present two examples of PKAT, paraconsistent sets and relations, which are parametric on arbitrary twisted structures resulting from the direct product of the relevant lattices. These examples serve to illustrate how PKAT handles computations with a paraconsistent reasoning.

\medskip

\textbf{Paper structure.} Subsection~\ref{subsec:prel} recalls the definition of \texttt{KA} and \texttt{KAT}. Section~\ref{sec:PLTS} revisits the definition of PLTS parametric on a class of residuated lattices \cite{lsfa22} and establishes some new properties that will prove useful in the sequel. Section~\ref{sec:Katp} introduces a variant of \texttt{KAT} for a paraconsistent context (\texttt{PKAT}) where programs and tests accommodate inconsistencies and vagueness. Additionally, in Section~\ref{sec:Katp} we present the details of two new algebraic structures that form a \texttt{PKAT}: paraconsistent sets $\Set(\T)$ and paraconsistent relations $\Rel(\K,\T)$, which are parametric over fixed twisted structures $\K$ and $\T$, respectively. Finally, Section~\ref{sec:conclusion} concludes and points out a number of topics for future research.

\subsection{Preliminaries}\label{subsec:prel}

\begin{definition} \label{kleenealg} \cite{Kozen97}
A Kleene algebra \texttt{(KA)} is an algebraic structure 
$(K,+,\cdot,^\star,0,1)$ 
satisfying the axioms~\eqref{1}-\eqref{15} below. The order of precedence of the operators is $\,^\star \, >\, \cdot \,>\, +$. Thus, $p+q\cdot r^\star $ should be parsed as $p+(q\cdot(r^\star))$.
\vspace{-0.5cm}
\begin{center}\begin{minipage}{0.5\linewidth}
    \begin{gather}
        p+(q+r) \,=\, (p+q)+r  \label{1} \\
         p+q\,=\, q+p \label{2}\\
         p+0\,=\,p \label{3}\\
         p+p\,=\, p \label{4} \\
         p\cdot (q \cdot r)\,=\,(p \cdot q) \cdot r  \label{5}\\
         1\cdot p \,=\,  p \cdot 1\,=\, p \label{6}\\
         p\cdot (q+r)\,=\, p\cdot q+p\cdot r \label{7}
    \end{gather}
\end{minipage}
\begin{minipage}{0.455\linewidth}
    \begin{gather}
        (p+q)\cdot r\,=\, p\cdot r +q \cdot r \label{8}\\
         0 \cdot p\,=\,p\cdot 0\,=\,0\label{9} \\
         1+p \cdot p^\star \,=\, p^\star \label{10} \\
         1+ p^\star \cdot p \,=\, p^\star \label{11} \\
         p\cdot r \,\leq \, r  \, \longrightarrow  \, p^\star \cdot r \,\leq\, r  \label{14}\\
         r \cdot p \,\leq\, r  \, \longrightarrow  \, r \cdot p^\star \,\leq \, r \label{15}
    \end{gather}
\end{minipage}
\end{center}
where $\leq$ refers to the natural partial order on $K$, that is, $p \leq q$ if and only if  $p+q= q$.
\end{definition}
Axioms \eqref{1}-\eqref{9} establish $(K,+,\cdot,^\star,0,1)$ as an idempotent semiring, while axioms \eqref{10}-\eqref{15} say that $^\star$ is like the reflexive transitive closure on binary relations \cite{10.5555/557365}.

\begin{definition}\cite{10.5555/557365}
    A \emph{Kleene algebra with tests} \texttt{(KAT)} is a two-sorted algebra
    $$(K,B,+,\cdot, ^\star,^-,0,1)$$
    such that $(K,+,\cdot,\,^\star, 0,1)$ is a Kleene algebra, $(B,+,\cdot,\,^-,0,1)$ is a Boolean algebra and $B \subseteq K$.
\end{definition} 
The unary operator $\,^-$ is defined only on $B$ which elements are called tests. We reserve the letters $p,\, q,\, r,\, s $ for arbitrary elements of $K$ and $\alpha,\,\beta,\,\gamma$ for tests. In summary, a \texttt{KAT} satisfies axioms \eqref{1}-\eqref{15} and the following for any tests:
\vspace{-0.5cm}
\begin{center}\begin{minipage}{0.5\linewidth}
    \begin{gather}
    \alpha + (\beta \cdot \gamma)\,=\, (\alpha+\beta) \cdot (\alpha+\gamma)   \label{213} \\
    (\alpha \cdot \beta) + \gamma\,=\, (\alpha+\gamma) \cdot (\beta+\gamma) \label{215}\\ 
    \alpha \cdot \beta \,=\, \beta \cdot \alpha  \label{214} \\
    \alpha \cdot \alpha \,=\, \alpha \label{216} 
    \end{gather}
\end{minipage}
\begin{minipage}{0.455\linewidth}
    \begin{gather}
        \overline{\overline{\alpha}}\,=\, \alpha \label{217} \\
         \alpha+1\,=\, 1 \label{218} \\
         \alpha \cdot \overline{\alpha}\,=\,0  \label{219}\\
         \alpha + \overline{\alpha}\,=\,1 \label{220}
    \end{gather}
\end{minipage}
\end{center}
Axioms~\eqref{1}-\eqref{15} pertain to the fact that $(K,+,\cdot,^\star, 0,1)$ is a \texttt{KA}. While axioms~\eqref{213}-\eqref{220} pertain to the fact that  $(B,+,\cdot,^-,0,1)$ is a Boolean algebra \cite{10.5555/557365}. 

\begin{example}\label{ex:binrel}
    Semantically, programs are represented as binary relations over a set of states $X$ and a test $\alpha$ is interpreted as a subset of the identity relation, comprising all pairs $(x,x)$ such that $\alpha$ holds at state $x$. Hence,the family of binary relations on a set $X$ is a \texttt{KAT} with operations defined as follows.
    \vspace{-0.3cm}
    \begin{center}\begin{minipage}{0.35\linewidth}
        \begin{align*}
            &0 \coloneqq \emptyset   \\
            &1 \coloneqq \{(u,u)\,|\,u \in X\}  \\
            &\overline{\alpha} \,= \, 1 \setminus \alpha
        \end{align*}
    \end{minipage}
    \begin{minipage}{0.5\linewidth}
        \begin{align*}
             &  R + R' \coloneqq R \cup R'  \\
             &R \cdot R' \coloneqq \{(u,w)\,|\, \exists v  (u,v)\in R \land  (v,w) \in R'\} \\
             &R^\star \coloneqq \underset{n \geq 0}{\bigcup}R^n\,=\, \text{reflexive transitive closure of $R$} 
        \end{align*}
    \end{minipage}
    \end{center}
where $R^0\coloneqq \{(u,u)\,|\,u \in X\}$ and $R^{n+1}\,=\,  R \cdot R^n$.
\end{example}

This extension of \texttt{KA} with a Boolean algebra results in an algebraic model to capture program behavior and assertions. Hence, conditionals and \textbf{while} loops found in programming can be defined in terms of the regular operators.
\begin{gather*}
    \textbf{if $\alpha$ then $p$ else $q$} \;\overset{\text{def}}{=}\; \alpha \cdot p + \overline{\alpha} \cdot q \\
    \textbf{while $\alpha$ do $p$}\;\overset{\text{def}}{=}\; (\alpha\cdot p)^\star\,     \cdot \, \overline{\alpha}
\end{gather*}

\section{Paraconsistent transition systems}\label{sec:PLTS}
The notion of \emph{paraconsistent transition systems}, abbreviated to PLTS, was introduced in \cite{lsfa22}. These systems' transitions involve two weights: a \emph{positive} and a \emph{negative} that characterize each transition in opposite ways. The positive weight represents the evidence of its presence, while the negative weight represents the evidence of its absence.
Furthermore, following the line of research outlined in \cite{Bou}, a residuated lattice over a set $A$ of possible truth values is adopted. This allows the transitions of PLTS to be represented by pairs of weights $(\doublet, \doublef) \in A \times A$. Thus, all the relevant constructions of PLTS are parametric in a class of residuated lattices and admit different instances according to the truth values domain $A$ that better suits each concrete problem.

To exemplify, suppose that weights for both transitions come from a residuated lattice over the real interval $[0,1]$.
 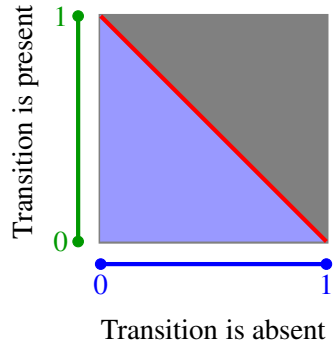
\begin{figure}[H]\begin{center}
\begin{tikzpicture}
\draw[gray, ultra thick] (0,0) rectangle (3,3);
\draw[black!40!green, ultra thick] (-0.3,0) -- (-0.3,3) ;
\draw[blue, ultra thick] (0,-0.3) -- (3, -0.3);
\draw[] (-1,3.3) node[anchor=east, rotate=90] {Transition is present};
\draw[] (1.5,-0.9) node[anchor=north] {Transition is absent};
\filldraw[blue] (0,-.3) circle (2pt) node[anchor=north] {$0$};
\filldraw[blue] (3,-0.3) circle (2pt) node[anchor=north] {$1$};
\filldraw[black!40!green] (-0.3,0) circle (2pt) node[anchor=east] {$0$};
\filldraw[black!40!green] (-0.3,3) circle (2pt) node[anchor=east] {$1$};
\filldraw[gray] (3,0) -- (3,3) -- (0,3);
\filldraw[white!60!blue] (0,0) -- (0,3) -- (3,0);
\draw[red, ultra thick] (0,3) -- (3,0);
\end{tikzpicture}\end{center}
\caption{The vagueness-inconsistency square}\label{fig:square}
\end{figure}
\noindent These pairs of weights express different behaviors:
\begin{itemize}
\item \emph{inconsistency}, when the positive and negative weights  are contradictory,i.e., they sum to a value greater than $1$, this corresponds to the upper triangle in Figure~\ref{fig:square}, filled in grey.
\item \emph{vagueness}, when the sum is less than $1$,  corresponding to the lower, periwinkle  triangle in Figure~\ref{fig:square}.
\item \emph{consistency}, when the sum is exactly $1$, that is the evidence degrees that enforce and prevent a transition from occurring are complementary, corresponding to the red line in Figure~\ref{fig:square}.
\end{itemize}

We will consider a class of residuated lattices $\pmb{A} = \langle A, \sqcap, \sqcup, 1, 0, \rightarrow \rangle$ over a set $A$, bounded by a maximal element $1$ and a minimal element $0$ and where the lattice meet ($\sqcap$) and the monoidal composition ($\odot$) coincide. Such lattices are commonly referred to as \emph{Heyting algebras} in the literature, and also known as pseudo-Boolean algebras \cite{Rasiowa63}. The adjunction property is stated as $a \sqcap b \leq c \; \text{ if and only if } \; b \leq a \to c$. Finally, we will require that the Heyting algebras in the sequel be \emph{complete}, i.e., every subset of $A$ has both a greatest lower bound and a least upper bound.

\begin{example} The following lattices are complete Heyting algebras:
    \begin{itemize}
    \item the Boolean algebra $\pmb{2}= \langle \{ 0,1\},\land,\lor,1, 0,\rightarrow \rangle$
    \item the \L{}ukasiewicz three-valued algebra $\pmb{3}= \langle \{ \top, u, \bot \},\land_3,\lor_3,\,\top ,\bot,\rightarrow_3 \rangle$, where 
    \begin{center}
    \begin{tabular}{c|ccc}
    $\;\land_{3}\;$ & $\,\bot\,$ & $\,u$ & $\,\top\,$\tabularnewline\hline $\bot$ & $\bot$ &$\bot$ & $\bot$\tabularnewline$u$ & $\bot$ & $u$ & $u$\tabularnewline$\top$ & $\bot$ & $u$ & $\top$\tabularnewline\end{tabular}
    \quad
    \begin{tabular}{c|ccc}
    $\;\lor_{3}\;$ & $\,\bot\,$ & $\,u\,$ & $\top$\tabularnewline\hline$\bot$ & $\bot$ & $u$ & $\top$\tabularnewline$u$&$u$ & $u$ & $\top$\tabularnewline$\top$ & $\top$ &$\top$ & $\top$\tabularnewline
    \end{tabular}
    \quad
    \begin{tabular}{c|ccc}
    $\;\rightarrow_{3}\;$ & $\,\bot\,$ & $\,u\,$ & $\,\top\,$\tabularnewline\hline 
    $\bot$ & $\top$ & $\top$ & $\top$\tabularnewline
    $u$ & $u$ & $\top$ & $\top$\tabularnewline$\top$ & $\bot$ & $u$ & $\top$\tabularnewline\end{tabular}
    \end{center}
    The truth value $u$ stands for ``unknown'' and assumes different notations in the literature, such as $ \nicefrac{1}{2}$ or $\#$, which are interpreted as ``possibility'' or ``indeterminacy'' \cite{Priest07}.
    \item $\pmb{G}= \langle [0,1],\min,\max,0, 1,\rightarrow \rangle$, with $a \rightarrow b\,=\,1$, if $a\leq b$ and $a \rightarrow b\,=\,b$ otherwise.
    \end{itemize}
\end{example}

The following lemma combines \cite[Lemma~1]{Bisparaconsistent} and delineates several essential properties of complete residuated lattices, which we will resort to prove some results in this paper. For detailed proofs of Properties~\eqref{dist-and}-\eqref{sqcap-monotone} and Properties~\eqref{prop1-13}-\eqref{prop1-14}, readers are referred to \cite{Bisparaconsistent} and \cite{Bou}, respectively. 

\begin{lemma} Let $\pmb{A}$ be an complete Heyting algebra over a non empty set $A$. The following properties hold, for any $a,\, a', \, b, \, b' \in A$
\begin{gather}
    a \sqcap (b \sqcup b') \, =\, (a \sqcap b) \sqcup (a \sqcap b')  \label{dist-and}\\
    a \sqcup (b \sqcap b') \, =\, (a \sqcup b) \sqcap (a \sqcup b')  \label{dist-or}\\
    a \leq a' \text{ and } b \leq b'  \text{ implies } a \sqcap b \leq a'\sqcap b' \label{sqcap-monotone}\\
    a \sqcap \bigg( \, \underset{i \in I}{\bigsqcup} \, b_i \, \bigg) \,=\, \underset{i \in I}{\bigsqcup} \, (a \sqcap b_i) \label{prop1-13} \\
    \bigg( \, \underset{i \in I}{\bigsqcup} \, a_i \, \bigg) \sqcap b \,=\, \underset{i \in I}{\bigsqcup} \, (a_i  \sqcap b) \label{prop1-14}
\end{gather}
where $I$ is a (possibly infinite) index set.
\end{lemma}

Before providing a formal definition of PLTS, it is important to define the notion of a twisted structure. Initially introduced in Kracht's seminal work \cite{Kracht1998-KRAOEO}, a twisted structure results from the construction of the direct product of a residuated lattice $\pmb{L}$ and its order-dual $\pmb{L}^\partial$. This resulting lattice naturally possesses an involution\footnote{Informally speaking, the operator $\sslash$ will denote a involutive negation in this paper and in essence it interchanges the positive and negative weight within a pair. See Example~\ref{ex1:PLTS} for a concrete example.} given by 
$$\sslash(a,a')\,=\, (a',a)$$
for all $(a,\,a') \in \pmb{L} \times \pmb{L}^\partial$.
Since its introduction, numerous authors have explored extensions of this structure by imposing additional properties on the residuated lattice \cite{BusanicheCR14,BusanicheGM22,Castiglioni08}. Twisted structures play a fundamental role in the context of PLTS, as documented in prior work \cite{NCL22,lsfa22,TASE23,FSEN23}, by enabling the computation of pairs of truth weights. 

\begin{definition}\label{def:twisted-alg} Let \pmb{A}$=\langle A, \sqcap, \sqcup,1,0,\rightarrow \rangle$ be a complete Heyting algebra.
Its corresponding $\pmb{A}$-twisted structure is denoted by $\A=\langle \, A \times A, \sqdoublecup , \sqdoublecap ,\sslash,(0,1),(1,0) \, \rangle $ and defined for any pair in $A \times A$ as:
\begin{gather*}
    \sslash (a,b)\,=\, (b,a) \\
    (a,a') \sqdoublecap (b,b') = (a \sqcap b, a' \sqcup b')\\
    (a,a') \sqdoublecup (b,b') = (a \sqcup b, a' \sqcap b')
\end{gather*}
The order in $\pmb{A}$ is lifted to $\A$ as $(a,a') \preccurlyeq (b,b') \text{ iff } a \leq b \text{ and } a' \geq b'$.
\end{definition}
As before, 
parentheses will be frequently omitted, reserving their use for enhancing the readability and clarity of certain expressions.

In the sequel, we will frequently refer to the next Lemma, which presents key properties of the operators defined above.  Although most of these properties are evident, since the product of two complete, universally distributive lattices is itself a complete, universally distributive lattice, we explicitly state them for ease of reference in later sections.  Lemma~\ref{lem:twistprop} essentially states that the twisted-structure in Definition~\ref{def:twisted-alg} has the structure of a De Morgan algebra —i.e., a bounded distributive lattice in which $\sslash$ is involutive and satisfies De Morgan’s laws (c.f. \cite{Kalman1958}). Additionally, it also states that the lattice $\langle A, \times A, \sqdoublecap, \sqdoublecup \rangle$ is a quantale, that is, a complete lattice equipped with an associative operation $\sqdoublecap$ that satisfies distributive properties. This aligns with the inspiration for twisted structures drawn from Chu's work in category theory and its application to quantales \cite{Tsinakis2006}.

\begin{lemma}\label{lem:twistprop}
    Let $\A$ be a twisted structure, as defined in Definition~\ref{def:twisted-alg}. The following properties hold for any $(a,a'),\, (b,b'),\,(c,c') \in A \times A$.
    \vspace{-0.5cm}
    \begin{gather}
         \sslash  \sslash (a,a') \,=\, (a,a') \label{prop2-16}\\
         (a,a') \sqdoublecup (1,0) \,=\,(1,0) \label{prop2-17} \\
         (a,a') \sqdoublecup (0,1)\,=\, (a,a') \label{prop2-3} \\
         (a,a') \sqdoublecup (a,a') \,= \,(a ,a') \label{prop2-4} \\
         (a,a') \sqdoublecap (a,a')\,=\,(a,a') \label{prop2-15}\\
         (a,a') \sqdoublecup (b,b')\,=\, (b,b') \sqdoublecup (a,a') \label{prop2-2} \\
         (a,a') \sqdoublecap (b,b') \,=\, (b,b') \sqdoublecap (a,a') \label{prop2-11}\\
         \sslash ((a,a') \sqdoublecup (b,b')) \,=\, \sslash(a,a') \sqdoublecap \sslash(b,b') \label{prop2-10}\\
         \sslash ((a,a') \sqdoublecap (b,b')) \,=\, \sslash(a,a') \sqdoublecup \sslash(b,b') \label{prop2-12} \\
         (a,a') \sqdoublecap (1,0) \,=\, (1,0) \sqdoublecap (a,a')\,=\,(a,a') \label{prop2-6} \\  
         (a,a') \sqdoublecap (0,1)\,=\, (0,1) \sqdoublecap (a,a')\,=\,(0,1) \label{prop2-9} \\
        (a,a') \sqdoublecup \bigg((b,b') \sqdoublecup (c,c')\bigg) \,=\, \bigg((a,a') \sqdoublecup (b,b')\bigg) \sqdoublecup (c,c') \label{prop2-1} \\
         (a,a') \sqdoublecap \bigg((b,b') \sqdoublecap (c,c')\bigg) \,=\, \bigg((a,a') \sqdoublecap (b,b')\bigg) \sqdoublecap (c,c') \label{prop2-5}  \\
         (a,a') \sqdoublecap \bigg( (b,b') \sqdoublecup (c,c')\bigg) \,=\, \bigg( (a,a') \sqdoublecap (b,b') \bigg) \sqdoublecup \bigg( (a,a') \sqdoublecap (c,c')\bigg) \label{prop2-7} \\
         (a,a') \sqdoublecup \bigg( (b,b') \sqdoublecap (c,c')\bigg) \,=\,    \bigg((a,a') \sqdoublecup  (b,b') \bigg)  \sqdoublecap \bigg( (a,a') \sqdoublecup (c,c')\bigg) \label{prop2-18} 
    \end{gather}
\end{lemma}

\begin{proof} 
\begin{itemize} 
    \item[]
    \item Property~\eqref{prop2-16} follows immediately since $\sslash$ is involutive.
     \item Property~\eqref{prop2-17} and Property~\eqref{prop2-3} are a consequence of $0$ and $1$ being the the least and greatest element of set $A$, respectively. Hence, 
     \begin{gather*}
         (a,a') \sqdoublecup (1,0) \,=\,( a \sqcup 1, a' \sqcap 0)\,=\,(1,0) \\
         (a,a') \sqdoublecup (0,1)\,=\, (a \sqcup 0, a' \sqcap 1)\,=\, (a,a')
     \end{gather*}
    \item Property~\eqref{prop2-4} is a consequence of operators $\sqcap$ and $\sqcup$ being idempotent, $(a,a') \sqdoublecup (a,a') \, = \, (a \sqcup a, a' \sqcap a')\,=\, (a, a')$. Property~\eqref{prop2-15} follows similarly.
    \item  Property~\eqref{prop2-2} is a consequence of operators $\sqcap$ and $\sqcup$ being commutative. Consequently, it follows $(a,a') \sqdoublecup (b,b')  \,=\, (a \sqcup b, a' \sqcap b') \,=\, (b \sqcup a, b' \sqcap a')  \,=\, (b,b') \sqdoublecup  (a,a')$. The proof for Property~\eqref{prop2-11} follows similarly. 
    \item The proof of Property~\eqref{prop2-10} follows directly by the definition of the operators,  $\sslash ((a,a') \sqdoublecup (b,b'))\,=\,(a' \sqcap b', a \sqcup b)\,=\,(a',a) \sqdoublecap (b',b) \,=\, \sslash(a,a') \sqdoublecap \sslash(b,b')$. Similarly, it is possible to prove Property~\eqref{prop2-12}.
    \item Property~\eqref{prop2-6} since $0$ and $1$ are the least and greatest element of $A$, respectively, it follows $(1,0) \sqdoublecap (a,a') \,=\,( 1 \sqcap a, 0 \sqcup a') \,=\, (a,a')$ and by \eqref{prop2-11} it follows $ (a,a') \sqdoublecap (1,0)\,=\,(a,a')$. The proof of Property~\eqref{prop2-9} follows similarly.
    \item To prove Property~\eqref{prop2-1} note that, 
    \begin{align*}
        &(a,a') \sqdoublecup ((b,b') \sqdoublecup (c,c'))
            = (a,a') \sqdoublecup (b \sqcup c, b' \sqcap c')
            = (a \sqcup (b \sqcup c), a' \sqcap (b' \sqcap c'))  \\
        &((a,a') \sqdoublecup (b,b')) \sqdoublecup (c,c') =(a \sqcup b, a' \sqcap b') \sqdoublecup (c,c') =((a \sqcup b) \sqcup c, (a' \sqcap b') \sqcap c')   
    \end{align*}
    Since $\sqcup$ and $\sqcap$ are associative it follows that $a \sqcup (b \sqcup c)=(a \sqcup b) \sqcup c$ and $a' \sqcap (b' \sqcap c')=(a' \sqcap b') \sqcap c'$. Therefore, operator $\sqdoublecup$ is also associative. Similarly, it is possible to prove Property~\eqref{prop2-5},i.e., $\sqdoublecap$ is associative.
    \item Property~\eqref{prop2-7} is a consequence of Property~\eqref{dist-and} and \eqref{dist-or}.
    \begin{align*}
            (a,a') \sqdoublecap \bigg((b,b') \sqdoublecup (c,c')\bigg) 
            =& (a,a') \sqdoublecap (b \sqcup c,b' \sqcap c')\\
            =& (a \sqcap (b \sqcup c),a' \sqcup (b' \sqcap c')) \\
            =& ((a \sqcap b) \sqcup (a \sqcap c), (a' \sqcup b') \sqcap (a' \sqcup c')) \\
            =& (a \sqcap b,a' \sqcup b') \sqdoublecup (a \sqcap c, a' \sqcup c') \\
            =& \bigg((a,a') \sqdoublecap (b,b')\bigg) \sqdoublecup \bigg((a,a') \sqdoublecap (c,c') \bigg)
    \end{align*}
\end{itemize}
\end{proof}

Finally, we present the definition of \emph{paraconsistent transition systems} parametric in a class of residuated lattices over a set $A$ of truth values.
\begin{definition}\label{def:systems}
  A paraconsistent transition system, abbreviated to PLTS, is a tuple $M=\model$ such that:
  \begin{itemize}
    \item $W$ is a non empty set of states,
    \item $R \colon W \times W \rightarrow A \times A$ is a paraconsistent accessibility relation. For any pair of states $(w_1,w_2)\in W \times W$ relation $R$ assigns a pair $(\doublet , \doublef) \in A \times A $ where $\doublet$ represents the evidence of the transition from $w_1$ to $w_2$ occurring and $\doublef$ represents the evidence of being prevented from occurring.
    \item $V:W\times \text{Prop} \rightarrow A \times A$ is a valuation function, that assigns to a proposition symbol $p$ at a given state $w$ a pair $(\doublet,\doublef) \in A \times A$ such that $\doublet$ is the evidence of $p$ holding at $w$ and $\doublef$ the evidence of not holding
  \end{itemize}
\end{definition}

\begin{example}\label{ex1:PLTS}
     Consider the residuated lattice $\pmb{3}$ and the set of proposition symbols $\{p\}$. The following model $M=(\{w_1,w_2\},R,V)$ is a PLTS where 
     \begin{center}\begin{minipage}{0.35\linewidth}
     \begin{tikzpicture}
             \node(A){$w_1$};
             \node(B)[right=2 cm of A]{$w_2$};
             \draw[->] (A) to [bend left=15] node [above]{$(\top,\bot)$} (B);
             \draw[->] (B) to [bend left=15] node [below]{$(\top,u)$} (A);
         \end{tikzpicture}
    \end{minipage}
    \begin{minipage}{0.4\linewidth}
        \begin{align*}
            V \,\colon \, W \times \{p\} &\rightarrow \{ \top, u, \bot \}\times  \{ \top, u, \bot \} \\
            (w_1,p)&\mapsto (\top,\bot)\\
            (w_2,p)& \mapsto (u,\bot)  
        \end{align*}
    \end{minipage}
    \end{center}
    We consider, following \L{}ukasiewicz , that the truth value $u$ is an intermediate value between $\top$ and $\bot$ \cite{Priest07}. Hence, a natural ``ordering''  of the three values is $\bot\,\leq \, u\, \leq \, \top$. In this example, each weight takes values in the set $\{\bot, u,\top\}$, making it possible to represent the resulting lattice of doing the direct product of the three-valued lattice and its  order-dual.
    
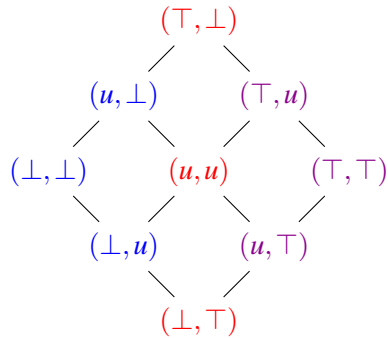
\begin{center}\begin{tikzpicture}
          \node[red1] (max) at (0,2) {$(\top,\bot)$};
          \node[blue1] (a) at (-1,1) {$(u,\bot)$};
          \node[magenta1] (c) at (1,1) {$(\top,u)$};
          
          \node[blue1] (d) at (-2,0) {$(\bot, \bot)$};
          \node[red1] (e) at (0,0) {$(u,u)$};
          \node[magenta1] (f) at (2,0) {$(\top,\top)$};
          
          \node[blue1] (g) at (-1,-1) {$(\bot,u)$};
          \node[magenta1] (h) at (1,-1) {$(u,\top)$};
          \node[red1] (min) at (0,-2) {$(\bot,\top)$};
          \draw (max) -- (a) -- (e) -- (c) -- (max);
          \draw (a) -- (d) -- (g) -- (e) -- (h) -- (f) -- (c);
          \draw (g) -- (min) -- (h);
    \end{tikzpicture}\end{center}
    All pairs marked in red represent consistent information. The pairs on the left, shown in blue, represent vague information, while those on the right, highlighted in magenta, represent inconsistent information.

    Let us observe the intuition behind some of the operators in the twisted structure defined in Definition~\ref{def:twisted-alg}. For example, $V(w_2,p)\,=\,(u,\bot)$ indicates that at state $w_2$, there is uncertain evidence ($u$) that $p$ holds and minimal evidence ($\bot$) that $p$ does not hold. Similarly, one could say that at state $w_2$, there is minimal evidence  that the negation of p holds and uncertain evidence that it does not hold. This reflects the intuition behind the operator $\sslash$, which acts as an involutive negation by switching the positive and negative weights.

    Considering the valuations of $p$ at states $w_1$ and $w_2$, it is possible to compute the evidence of $p$ holding or not at both states. The certainty that $p$ holds in both states is equal to the certainty that it holds in each state ($\top \sqcap u$). Similarly, the certainty that $p$ does not hold in both states is  equal to the certainty that it does not hold in either state ($\bot \sqcup \bot$). This captures the intuition behind the operator $\sqdoublecap$.
    
    The interested reader is referred to \cite{FSEN23} for a more detailed exploration of the paraconsistent logic underlying PLTS.
 \end{example}

\section{Paraconsistent Kleene Algebra with tests}\label{sec:Katp}
The approach presented in this work aims to reason about program executions in a paraconsistent manner, where executions and tests may involve vagueness as well as inconsistencies. Consequently, rather than yielding a bivalent outcome as in traditional \texttt{KAT}, the outcome is graded by a pair of weights, one weight indicates the evidence for execution and the other the evidence for failure. Such framework entails the need to weaken the Boolean subalgebra of \texttt{KAT} which leads to the following variant:
\begin{definition}
    A \emph{paraconsistent Kleene algebra with tests} (\texttt{PKAT}) is a tuple
    $$(K,T,+,\cdot, \,^\star, \,^-, 0,1)$$
    such that 
    $(K,+,\cdot,\,^\star, 0,1)$ is a Kleene algebra, $(T,+,\cdot,\,^-,0,1)$ satisfies axioms~\eqref{213}-\eqref{218} and $B \subseteq T$.
    Relation $\leq$ is induced by $+$, that is, $p \leq q$ iff $p+q=q$.
\end{definition}
A key aspect of the \texttt{KAT} axiomatization lies in axioms~\eqref{219} and \eqref{220}, which informally express the principles of non-contradiction and excluded middle, respectively. These two principles play a significant role in the philosophy of paraconsistency, which rejects the principle of non-contradiction, and in intuitionism, which does not assume the principle of the excluded middle. The notion of \texttt{PKAT} takes in consideration these philosophies and consequently forms a weakened version of \texttt{KAT} by only rejecting axioms \eqref{219} and \eqref{220}. Hence, 
\begin{theorem}
    Any \texttt{KAT} is a \texttt{PKAT}.
\end{theorem}
\begin{proof} By definition, any \texttt{KAT} satisfies axioms~\eqref{1}-\eqref{218}. Thus, trivially any \texttt{KAT} is a \texttt{PKAT}.
\end{proof}

The weakening discussed in this paper generalizes Boolean algebras, in that Boolean algebras are precisely those algebras that satisfy both the principle of non-contradiction \eqref{219} and the principle of the excluded middle \eqref{220}. This generalization is similar to that presented by Heyting algebras \cite{Borceux94}, where any Boolean algebra is a Heyting algebra that satisfies the principle of the excluded middle. In fact, the weakening presented in this paper is a stronger generalization than Heyting algebras, as any algebra failing to satisfy both principles inherently does not satisfy the principle of the excluded middle. A potential implication of this observation is that the discussed weakening of the Boolean algebra could potentially serve as algebraic models for propositional paraconsistent logic, much like how Heyting algebras model propositional intuitionistic logic \cite{BJ05} and Boolean algebras model propositional classical logic. 

In the remaining of this section we introduce two examples of algebras where program executions and tests may encompass inconsistencies and vagueness. Consequently, these algebras can be formalized as \texttt{PKAT}. To achieve this, we refer back to Definition~\ref{def:twisted-alg}, which establishes that for any complete Heyting algebra $\pmb{K}$ over a non-empty set $K$ of possible truth values, its corresponding twisted structure $\K$ allows for the computation of pairs of truth values $K \times K$.

\begin{remark}
    Given sets $T$ and $W$, we denote by $(T \times T)^W$ the set of functions $W \rightarrow (T \times T)$.
\end{remark}

\begin{definition}\label{def:setp}
    Let $W$ be a set and $\pmb{T}$ be a complete Heyting algebra over a non empty set of truth values $T$. The algebra of paraconsistent sets over the twisted-structure $\T$ is 
    $$\Set(\T) \,=\, \langle (T \times T)^W , (T \times T)^W,  + , \cdot , \,^\star,\,^-, \oslash,\varUpsilon \rangle $$
    For any paraconsistent sets over $W$, $\varphi, \psi \in (T \times T)^W$ and $w \in W$, operators are defined pointwise by
    \vspace{-0.5cm}
    \begin{center}\begin{minipage}{0.4\linewidth}
    \begin{gather*}
        \oslash(w) \,=\, (0,1) \\
        \varUpsilon(w) \,=\, (1,0) \\
        \overline{\varphi}(w) \,=\,  \sslash \varphi(w) 
    \end{gather*}
\end{minipage}
\begin{minipage}{0.4\linewidth}
    \begin{gather*}
        (\varphi + \psi)(w) \, = \, \varphi(w) \sqdoublecup \psi(w) \\
        (\varphi \cdot \psi)(w) \,=\, \varphi(w) \sqdoublecap \psi(w) \\  
        (\varphi^\star)(w) \,=\, \underset{n \geq 0}{ \bigsqdoublecup } \, \varphi^n(w)     
    \end{gather*}
\end{minipage}
\end{center}
    with $\varphi^0(w)=\varUpsilon(w)$ and $\varphi^{n+1}(w)\,=\, (\varphi\cdot \varphi^n)(w)$. The values of paraconsistent sets $\varphi(w)$ and $\psi(w)$ are elements of $T \times T$, and constants $\oslash $ and $ \varUpsilon$ are the least and the greatest elements of $T \times T$, respectively. The partial order in paraconsistent sets $\varphi$ and $\psi$ in $(T \times T)^W$ is given by 
    $$ \varphi \subseteq \psi \text{ if and only if } \forall \, w \in W,\, \varphi(w) \preccurlyeq \psi (w)$$
\end{definition}

\begin{example}
    Consider the set $W=\{w_1,w_2\}$ and the three-valued residuated lattice $\pmb{3}\,=\,\{\top,u, \bot\}$. Let $\varphi, \, \psi \,\in \Set(\twist{\mathbf{3}})$ be two paraconsistent sets defined as
    \begin{equation*}
    \begin{aligned}[c]
    \varphi \colon \, & W \rightarrow \pmb{3} \times \pmb{3} \\
    & w_1 \mapsto (\top,u)\\
    & w_2 \mapsto (u,u) 
    \end{aligned}
    \quad  \; \; \quad 
    \begin{aligned}[c]
    \psi \colon \, & W \rightarrow \pmb{3} \times \pmb{3} \\
    & w_1 \mapsto (\top,\bot)\\
    & w_2 \mapsto (\top,u)
    \end{aligned}
    \end{equation*}
   It is possible to define with operator $\,^-$ two other paraconsistent sets denoted by $\overline{\varphi}, \, \overline{\psi} \,\in \Set(\twist{\mathbf{3}})$ defined as,
   \begin{equation*}
   \begin{aligned}[c]
    \overline{\varphi} \colon \, & W \rightarrow \pmb{3} \times \pmb{3} \\
    & w_1 \mapsto (u,\top)\\
    & w_2 \mapsto (u,u) 
    \end{aligned}
    \quad  \; \; \quad 
    \begin{aligned}[c]
    \overline{\psi} \colon \, & W \rightarrow \pmb{3} \times \pmb{3} \\
    & w_1 \mapsto (\bot,\top)\\
    & w_2 \mapsto (u,\top)
    \end{aligned}
    \end{equation*}
    Note that $\varphi \subseteq \psi$, while $\overline{\psi} \subseteq \overline{\varphi}$.
\end{example}

\begin{theorem}\label{twist-kleene}
    For any complete Heyting algebra $\pmb{T}$ over a non empty set $T$ of possible truth values, $\Set(\T)$ forms a \texttt{PKAT}.
\end{theorem}
\begin{proof} 
    For a fixed complete Heyting algebra $\pmb{T}$, by Definition~\ref{def:twisted-alg}, we define its twisted structure $\T$. We will prove that $\Set(\T)$ defined as in Definition~\ref{def:setp} forms a \texttt{PKAT}, that is, we show that axioms~\eqref{1}-\eqref{218} are satisfied. 
   \\ \underline{Axiom~\eqref{1}} by Property~\eqref{prop2-1}, $\varphi(w) \sqdoublecup (\psi(w) \sqdoublecup \phi(w)) \,
       =\, (\varphi(w) \sqdoublecup \psi(w)) \sqdoublecup \phi(w)$. Using the definition of $+$, it follows that $ (\varphi+(\psi+\phi))(w) \,=\, ((\varphi+\psi)+\phi)(w)$.
    \\ \underline{Axiom~\eqref{2}} by Property~\eqref{prop2-2} and the definition of $+$, $(\varphi+\psi)(w) \,=\, (\psi+\varphi)(w) $.
    \\\underline{Axiom~\eqref{3}} from Property~\eqref{prop2-3} it follows that, $(\varphi + \oslash)(w)\,=\,\varphi(w) \sqdoublecup \oslash(w)\,=\, \varphi(w) \sqdoublecup (0,1)\,=\,\varphi(w) $. 
    \\ \underline{Axiom~\eqref{4}} is immediate by Property~\eqref{prop2-4},  $(\varphi+\varphi)(w)\,=\,\varphi(w)$
    \\ \underline{Axiom~\eqref{5}} is a consequence of Property~\eqref{prop2-5}, $\varphi(w) \sqdoublecap (\psi(w) \sqdoublecap \phi(w)) \,=\,(\varphi(w) \sqdoublecap \psi(w)) \sqdoublecap \phi(w)$. Using the definition of $\cdot$, $ (\varphi \cdot (\psi \cdot \phi))(w)\,=\, ((\varphi \cdot \psi)\cdot \phi)(w) $.
   \\ \underline{Axiom~\eqref{6}} from Property~\eqref{prop2-6} it follows that
   $(\varUpsilon \cdot \varphi)(w)\,=\, (\varphi \cdot \varUpsilon)(w)\,=\, \varphi(w)$. 
   \\ \underline{Axiom~\eqref{7}} using Property~\eqref{prop2-7} and the definition of $+$ and $\cdot$ it follows that, $(\varphi\cdot (\psi +\phi))(w) \,=\, ( (\varphi \cdot \psi)+ (\varphi \cdot \phi))(w)$.
    \underline{Axiom~\eqref{8}} follows by Property~\eqref{prop2-11}.
    \\ \underline{Axiom~\eqref{9}} by Property \eqref{prop2-9} it follows that $(\oslash \cdot \varphi)(w)\,=\, (\varphi \cdot \oslash)(w)\,=\, \oslash(w)$.
    \\ \underline{Axiom~\eqref{10}} can be derived as follows, 
    \begin{align*}
            (\varUpsilon+(\varphi \cdot \varphi^\star))(w)
            =& \varUpsilon(w) \sqdoublecup \bigg(\varphi(w) \sqdoublecap \underset{n \geq 0}{ \bigsqdoublecup}\, \varphi^n \bigg)\\
            =& \varphi^0(w) \sqdoublecup \bigg(\varphi(w) \sqdoublecap (\varphi^0(w) \sqdoublecup \varphi(w) \sqdoublecup \varphi^2(w) \ldots ) \bigg) \\
            & \justif{using~\eqref{prop2-7} and defn. of $\,^\star$}
            =& \varphi^0(w)  \sqdoublecup ( \varphi(w) \sqdoublecup \varphi^2(w) \sqdoublecup \ldots )\\
            =& \varphi^\star(w) 
        \end{align*}
    Similarly, it is possible to show \underline{Axiom~\eqref{11}}.
    \\ \underline{Axiom~\eqref{14}} Let us start by assuming that $(\varphi \cdot \psi)(w)  \, \preccurlyeq \, \psi(w)$. Then, 
    \begin{align}
        (\varphi^\star \cdot \psi)(w) 
        =&  \bigg( \varphi^0(w) \sqdoublecup \varphi(w) \sqdoublecup \varphi^2(w) \sqdoublecup \ldots \bigg) \sqdoublecap \psi(w)  \notag \\
        &  \notag \justif{using \eqref{prop2-11}, \eqref{prop2-7} and  \eqref{prop2-6}} 
        =& \psi(w) \sqdoublecup (\varphi(w) \sqdoublecap \psi(w)) \sqdoublecup (\varphi^2(w)  \sqdoublecap \psi(w)) \sqdoublecup \ldots  \notag \\
        =& \psi(w) \sqdoublecup (\varphi \cdot \psi)(w) \sqdoublecup (\varphi^2 \cdot \psi)(w) \sqdoublecup \ldots \label{aux-1}
    \end{align}
    For any integer $n \geq 0$, $(\varphi^n \cdot \psi)(w)\,= \,\underset{n}{\underbrace{{\varphi \cdot \ldots \cdot \varphi}}} \cdot \psi(w)$. 
    Using the hypothesis $n$ times, it follows that $ (\varphi^n \cdot \psi)(w) \,\preccurlyeq\, \psi(w)$. Hence, by \eqref{aux-1} and given that $(\psi \cdot \psi)(w) \,\preccurlyeq \, \psi(w)$ it follows that 
    $$\psi(w) \sqdoublecup (\varphi \cdot \psi)(w) \sqdoublecup (\varphi^2 \cdot \psi)(w)\sqdoublecup \ldots \, \preccurlyeq \, \psi(w)$$
    Similarly, it is possible to show \underline{Axiom~\eqref{15}}. \underline{Axiom~\eqref{213}}  $(\varphi + (\psi \cdot \phi))(w)\,=\,((\varphi + \psi)\cdot (\varphi +  \phi))(w) $ follows by Property~\eqref{prop2-18}  and the definition of $\cdot$ and $+$.
    Similarly, \underline{Axiom~\eqref{215}} follows from Property~\eqref{prop2-4} and \eqref{prop2-18}. \underline{Axiom~\eqref{214}} results from Property~\eqref{prop2-11}, $(\varphi \cdot \psi)(w) \,=\, (\psi \cdot \varphi)(w)$. Finally, it is possible to show \underline{Axiom~\eqref{216}}, that is, $(\varphi \cdot \varphi)(w) \,=\, \varphi(w)$; \underline{Axiom~\eqref{217}}, that is, $\overline{\overline{\varphi}}(w)\,=\, \varphi(w)$ and 
    \underline{Axiom~\eqref{218}}, that is, $(\varphi+ \varUpsilon)(w)\,=\, \varUpsilon(w)$ follow directly from Property~\eqref{prop2-15}, \eqref{prop2-16} and \eqref{prop2-17}, respectively.
\end{proof}

The aim of the following definition is to explore paraconsistent programs similar to the well-known binary programs presented in Example~\ref{ex:binrel}. These paraconsistent programs involve computations that may exhibit vagueness or inconsistency. Thus potentially lending themselves to representation as PLTS with transitions and valuations weighted by pairs of weights, tailored to fit the specific problem domain. Paraconsistent relations are defined over a pair of states $W\times W$, where a test $\alpha$ can be interpreted at any state $w \in W$. Specifically, if a test $\alpha$ is evaluated by a pair $(\doublet,\doublef)$ at state $w$ , then 
$\doublet$ the evidence of the test holding and $\doublef$ measures the evidence of the test not holding at state $w$. Let's proceed to define the algebra of paraconsistent relations.
\begin{definition}\label{def:prel}
     Let $W$ be a set, $\pmb{K}$ and $\pmb{T}$ be complete Heyting algebras over a non empty set of truth values $K$ and $T$, respectively. The algebra of paraconsistent relations over $\K$ and $\T$  is defined as 
     $$\Rel(\K,\T) \,=\, \langle (K \times K)^{W \times W}, ( T \times T) ^{W \times W},  +, \cdot , \,^\star,  \,^-, \oslash,\varLambda \rangle $$
    where $(K \times K)^{W \times W}$ is the set of all paraconsistent relations over $W \times W$, i.e. functions $(W \times W) \to (K \times K)$. The elements of $(T \times T) ^{W \times W}$ are \emph{paraconsistent tests} $t$ such that $t(u,v)\,=\,(0,1)$ whenever $u \neq v$. 
    The operators of paraconsistent relations are defined pointwise by
    \vspace{-0.5cm}
    \begin{center}\begin{minipage}{0.4\linewidth}
    \begin{align*}
        &\oslash(u,v)\,=\,(0,1) \\
        &\varLambda(u,v)\,=\,
            \begin{cases}
            (1,0) & \text{if } u=v \\
            (0,1) & \text{otherwise}
            \end{cases} \\
        &\overline{t}(u,v)\,=\,\sslash t(u,v) 
    \end{align*}
    \end{minipage}
    \begin{minipage}{0.4\linewidth}
        \begin{align*}
            &(R + R')(u,v) \,=\, R(u,v) \sqdoublecup R'(u,v)\\
            &(R \cdot R')(u,v) \,=\, \underset{w \in W}{\bigsqdoublecup} \,(R(u,w) \sqdoublecap R'(w,v)) \\  
            &R^\star(u,v)\,=\, \underset{n\geq 0}{\bigsqdoublecup} \, R^n(u,v) 
        \end{align*}
    \end{minipage}
    \end{center}
    with $R^{n+1}(u,v)\,=\, (R \cdot R^n)(u,v)$ and $R^0(u,v)\,=\, \varLambda(u,v)$.
    The value of paraconsistent relations, $R(u,v)$ and $R'(u,v)$ are elements of $K \times K$, the value of $t(u,v)$ is an element of $T \times T$, and constants $\oslash,\, \varLambda$ are the least and the greatest elements of $T \times T$. The partial order $\subseteq$ for paraconsistent relations is given by
    $$ R \subseteq R' \; \text{ if and only if } \;\forall u,\,v \in W,\, R(u,v) \preccurlyeq R'(u,v)$$
\end{definition}

\begin{theorem}\label{theo:relp-Katp}
    Let $\pmb{K}$ and $\pmb{T}$ be complete Heyting algebras over a non empty set $K$ and $T$ of possible truth values such that $T \subseteq K$, $\Rel(\K,\T)$ forms a \texttt{PKAT}.
\end{theorem}
\begin{proof}
    Let $\pmb{K}$ and $\pmb{T}$ be complete Heyting algebras over set $K$ and $T$, respectively, such that $T \subseteq K$. It is possible to define the corresponding twisted structures $\K$ and $\T$ as described in Definition~\ref{def:twisted-alg}. We will prove that $\Rel(\K,\T)$ forms a \texttt{PKAT}, that is, axioms~\eqref{1}-\eqref{218} are satisfied. The satisfaction of axioms~\eqref{1}-\eqref{4} and \eqref{217} is similar to Theorem~\ref{twist-kleene}. Let us show the remaining.    \\
     \underline{Axiom~\eqref{5}} derives as follows
    \begin{align*}
        (R\cdot (R' \cdot R''))(w,v) 
        =& \underset{u \in W}{\bigsqdoublecup} \bigg( R(w,u) \sqdoublecap 
        \underset{t \in W}{\bigsqdoublecup} \bigg( R'(u,t) \sqdoublecap R''(t,v) \bigg) \bigg) \\
        =& \underset{u\in W}{\bigsqdoublecup}\, \underset{t \in W}{\bigsqdoublecup}  \bigg( R(w,u) \sqdoublecap R'(u,t) \sqdoublecap R''(t,v) \bigg)\\
        =& \underset{t \in W}{\bigsqdoublecup}  \bigg(  \underset{u \in W}{\bigsqdoublecup} \bigg( R(w,u) \sqdoublecap R'(u,t) \bigg) \sqdoublecap R''(t,v) \bigg) \\
        =& \underset{u \in W}{\bigsqdoublecup}  \bigg(  (R \cdot R')(w,t) \sqdoublecap R''(t,v) \bigg) \\
        =& ((R \cdot R')\cdot R'') (w,v)
    \end{align*}
    \underline{Axiom~\eqref{6}} by definition of $\cdot$, $(\varLambda \cdot  R) (w,v) = \underset{u \in W}{\bigsqdoublecup} \bigg( \varLambda(w,u) \sqdoublecap R(u,v) \bigg)$. \\
    For all $u \neq w$, $\varLambda(w,u)\,=\,(0,1)$ and by \eqref{prop2-9} it follows,  $\varLambda(w,u) \sqdoublecap R(u,v)\,=\,(0,1)$. Futhermore, by Property~\eqref{prop2-17} it follows that  
    $(\varLambda \cdot  R) (w,v) \,=\,  \varLambda(w,w) \sqdoublecap R(w,v) \,=\, (1,0) \sqdoublecap R(w,v)\,=\,R(w,v)$.
    \\ \underline{Axiom~\eqref{7}} proceeds as 
    \begin{align*}
        (R \cdot  (R'+R''))(w,v) 
        =& \underset{u \in W}{\bigsqdoublecup} (R(w,u) \sqdoublecap (R'(u,v) \sqdoublecup R''(u,v))) \\
        & \justif{using Property~\eqref{prop2-7}}
        =& \underset{u \in W}{\bigsqdoublecup} \bigg( \bigg(R(w,u) \sqdoublecap R'(u,v)\bigg) \sqdoublecup  \bigg( R(w,u) \sqdoublecap R''(u,v)\bigg) \bigg)\\
        =& \underset{u \in W}{\bigsqdoublecup}  \bigg(R(w,u) \sqdoublecap R'(u,v)\bigg) \, \sqdoublecup \, \underset{u \in W}{\bigsqdoublecup}   \bigg( R(w,u) \sqdoublecap R''(u,v)\bigg) \\
        =& (R \cdot  R')(w,v) + (R \cdot R'')(w,v)
    \end{align*}
    \underline{Axiom~\eqref{8}} follows similarly by Property~\eqref{prop2-11}.
    \\ \underline{Axiom~\eqref{9}} by  Property~\eqref{prop2-9} it follows that, 
    \begin{equation*}
        (\oslash \, \cdot\, R)(w,v) = \underset{u \in W}{\bigsqdoublecup} \bigg( \oslash(w,u) \sqdoublecap R(u,v) \bigg)  
        =\underset{u \in W}{\bigsqdoublecup} \bigg( (0,1) \sqdoublecap R(u,v) \bigg)  
        = (0,1)
        = \oslash(w,v) 
    \end{equation*}
    Consequently, since $\sqdoublecap$ is commutative, it follows $(R\, \cdot \,\oslash)(w,v) \,=\, \oslash(w,v)$.
    \\ Similar to \underline{Axiom~\eqref{11}}, \underline{Axiom~\eqref{10}} derives from,
    \begin{align*}
        (\varLambda+(R \cdot R^*))(w,v) 
        =& \varLambda(w,v) \sqdoublecup  \bigg(\underset{u \in W}{\bigsqdoublecup} R(w,u) \sqdoublecap \bigg( \underset{n \geq 0}{\bigsqdoublecup} R^n(u,v) \bigg)  \bigg)  \\
        & \justif{using Property \eqref{prop2-7}}
        =& \varLambda(w,v) \sqdoublecup  \bigg(  \underset{n \geq 0}{\bigsqdoublecup} \; \bigg( \underset{u \in W}{\bigsqdoublecup} R(w,u) \sqdoublecap  R^n(u,v) \bigg) \bigg) \\
        =& R^0(w,v) \sqdoublecup  \bigg(  \underset{n \geq 0}{\bigsqdoublecup}  (R \cdot R^n)(w,v)  \bigg) \\
        =& R^0(w,v) \sqdoublecup  \bigg(  \underset{n > 0}{\bigsqdoublecup}  \; R^{n}(w,v)  \bigg) \\
        =& R^*(w,v)
    \end{align*}
    \underline{Axiom~\eqref{14}}\, Let us assume that $(R \cdot R')(w,v) \,\preccurlyeq \, R'(w,v)$. Then,
    \begin{align*}
        (R^* \cdot R')(w,v) 
        =& \underset{u \in W}{\bigsqdoublecup}\; \bigg( \underset{n\geq 0}{\bigsqdoublecup}  R^n(w,u) \sqdoublecap R'(u,v) \bigg)\\ 
        =&\underset{n\geq 0}{\bigsqdoublecup} \; \bigg( \underset{u \in W}{\bigsqdoublecup} ( R^n(w,u) \sqdoublecap R'(u,v) ) \bigg) \\
        =&\underset{n\geq 0}{\bigsqdoublecup} \; (R^n \cdot R')(w,v) \\
        =& R'(w,v) \sqdoublecup (R\cdot R')(w,v) \sqdoublecup (R^2 \cdot R')(w,v) \sqdoublecup \ldots \\
        & \justif{Hypothesis}
        \preccurlyeq & R'(w,v)
    \end{align*}
    Similarly, it is possible to show \underline{Axiom~\eqref{15}}.
    \\ \underline{Axiom~\eqref{213}}  derives as,
    \begin{align*}
        (t + (t' \cdot t''))(u,v) 
        =& t(u,v) \sqdoublecup  \bigg(\underset{w \in W}{\bigsqdoublecup}(t'(u,w) \sqdoublecap t''(w,v)) \bigg)  \\
        &\justif{(step $\star$)}
       =& t(u,v) \sqdoublecup  \bigg( t'(u,v) \sqdoublecap t''(u,v) \bigg) \\
       & \justif{using Property~\eqref{prop2-18}} 
       =& \bigg( t(u,v) \sqdoublecup  t'(u,v) \bigg) \sqdoublecap \bigg( t(u,v) \sqdoublecup t''(u,v) \bigg) \\ 
        =&  (t+ t')(u,v)  \sqdoublecap (t + t'')(u,v) 
    \end{align*}
    (step $\star$) for any $w \in W$, $(t'(u,w) \sqdoublecap t''(w,v)) \neq (0,1)$, iff $(t'(u,w)\neq (0,1)$ and $t''(w,v) \neq (0,1))$. Thus, $(t'(u,w) \sqdoublecap t''(w,v)) \neq (0,1)$ only when $w=u$ and $w=v$. \\ Also note that,
    \begin{align*}
        ((t + t') \cdot ( t + t''))(u,v) 
        =& \underset{w \in W}{\bigsqdoublecup} \bigg( \bigg( t(u,w) \sqdoublecup t'(u,w) \bigg) \sqdoublecap \bigg(t(w,v) \sqdoublecup t''(w,v)\bigg) \bigg)  \\
        & \justif{(step $\star \star$)}
        =& \bigg( t(u,v) \sqdoublecup t'(u,v) \bigg) \sqdoublecap \bigg(t(u,v) \sqdoublecup t''(u,v)\bigg) \\
        =& (t+t')(u,v) \sqdoublecap (t+t'')(u,v) 
    \end{align*}
    (step $\star \star$) for any $w \in W$,
    \begin{gather*}
        \bigg( t(u,w) \sqdoublecup t'(u,w) \bigg) \sqdoublecap \bigg(t(w,v) \sqdoublecup t''(w,v)\bigg)  \,\neq \,(0,1)
    \end{gather*}
    if and only if $( t(u,w) \sqdoublecup t'(u,w) \neq (0,1) $ and $t(w,v) \sqdoublecup t''(w,v) \neq (0,1))$ Hence, $t(u,w)$, $t'(u,w)$, $t(w,v)$ and $t''(w,v)$ must all be different from $(0,1)$ which implies $u=w=v$. 
      \\ Therefore, we show that $ (t + (t' \cdot t''))(u,v)  \,=\,((t + t') \cdot ( t + t''))(u,v)  $. By Property \eqref{prop2-4} it is possible to show \underline{Axiom~\eqref{215}}.
     \\ \underline{Axiom~\eqref{214}} by definition,
    \begin{gather*}
        (t \cdot t')(u,v) 
        = \underset{w \in W}{\bigsqdoublecup} \bigg( t(u,w) \sqdoublecap t'(w,v) \bigg)  \\
        (t' \cdot t)(u,v) 
        =\underset{w \in W}{\bigsqdoublecup} \bigg( t'(u,w) \sqdoublecap t(w,v) \bigg)  
    \end{gather*}
    Whenever $w \neq u$ or $w\neq v$, by definition of test and by \eqref{prop2-9}, it follows that $ t(u,w) \sqdoublecap t'(w,v) =  t'(u,w) \sqdoublecap t(w,v) =(0,1) $. Hence, $t(u,w) \sqdoublecap t'(w,v) \neq (0,1)$ and $t'(u,w) \sqdoublecap t(w,v) \neq (0,1)$ only when $w = u$ and $w=v$. Thus, $$(t \cdot t')(u,v) ,=\, t(u,v) \sqdoublecap t'(u,v) \,=\, t'(u,v) \sqdoublecap  t(u,v)\,=\, (t' \cdot t)(u,v)$$
    \underline{Axiom~\eqref{216}} by definition of $\cdot$, $(t \cdot t)(u,v)
        \,=\, \underset{w\in W}{\bigsqdoublecup} \bigg( t(u,w) \sqdoublecap t(w,v) \bigg)$. \\ Since $ t(u,w) \sqdoublecap t(w,v)\, =\, (0,1)$ whenever $u\neq w $ or $ w\neq v$. Hence, by \eqref{prop2-3} and \eqref{prop2-15} $(t \cdot t)(u,v) \,=\, t(u,v) \sqdoublecap t(u,v)\,=\, t(u,v) $.
    \\ \underline{Axiom~\eqref{218}} By definition of $+$, $(t + \varLambda)(u,v)
    = t(u,v)+\varLambda(u,v)$. If $u=v$ then, $\varLambda(u,v)=(1,0)$ and by Property~\eqref{prop2-17} $(t + \varLambda)(u,v)
    = (1,0)= \varLambda(u,v)$. Otherwise, $t(u,v)=(0,1)$ and by Property~\eqref{prop2-3},  $(t + \varLambda)(u,v)=\varLambda(u,v)$. Hence,
    $(t + \varLambda)(u,v)=\varLambda(u,v)$.
\end{proof}

\section{Conclusion and future work}\label{sec:conclusion}
This paper contributes to an ongoing research agenda focused on paraconsistent transition systems (PLTS) and their logics. PLTS were initially introduced in \cite{lsfa22}, followed by the introduction of a logic to express their properties in \cite{NCL22} and the respective application to the analysis of quantum circuits was further discussed in \cite{BarbosaM23}. Subsequently, the algebra of constructors and abstractors for PLTS was discussed in \cite{TASE23}, which served as the foundation for a structured specification theory outlined in \cite{FSEN23}.

Here our focus is on developing an algebraic counterpart to this line of work. Hence, we take the initial steps towards introducing a variant of \texttt{KAT} to reason about vague and inconsistent computations and assertions. A similar roadmap for reasoning about fuzzy computations can be found in \cite{GMB19}. As in \cite{GMB19}, given that such assertions often take the form of tests, our approach lies in the modification of \texttt{KAT} that deals with properties of tests. The approach taken in this paper rejects the principle of non-contradiction and the principle of the excluded middle; consequently, some classical properties of Boolean algebra are lost. The resulting  \texttt{PKAT} can interpret computations entailing contradictions or vagueness.

A possible application of this work is in quantum circuits where a phenomenon known as \emph{decoherence} can occur. Such phenomenon is characterized by the loss of information from the circuit due to unwanted interaction with the environment. When the coherence time of a qubit is exceeded, there is an increasing probability that the circuit does not behave according to its design. Typically, qubit coherence is not specified exactly but is given as time intervals in the literature, corresponding to worst-case and best-case scenarios. In \cite{lsfa22}, it is proposed to use the two accessibility relations in PLTS to model these scenarios simultaneously. The minimum coherence time $t_{\text{min}}$ determines the negative weight, interpreted as the likelihood that the system evolves to a decoherent state, and the maximum coherence time $t_{\text{max}}$ determines the positive weight, interpreted as the likelihood that the system remains coherent.
With further considerations, it becomes feasible to translate quantum circuits into PLTS, as elaborated in \cite{BarbosaM23,lsfa22}.

The parametric approach adopted in this work is based on prior work documented in \cite{Bou}, which has been applied to formalize paraconsistent transition systems and their corresponding logics in \cite{NCL22,TASE23}. Given two residuated lattices $\pmb{K}$ and $\pmb{T}$ over a set of possible truth values $K$ and $T$ such that $T \subseteq K$, it is possible to define a twisted structure to operate on pairs $K \times K$. This structure allows for the introduction of two algebras in this paper: $\Set(\T)$ and $\Rel(\K,\T)$ parametric to the twisted structures. The main results demonstrate that both algebras $\Set(\T)$ and $\Rel(\K,\T)$ form a \texttt{PKAT}, Theorem~\ref{twist-kleene} and \ref{theo:relp-Katp}.

Since \texttt{KAT} provides a framework for reasoning about imperative programs in a (quasi) equational way, we aim to explore an encoding of propositional Hoare logic into \texttt{PKAT}. However, for such task it may be necessary to refine \texttt{PKAT} with additional properties and both the meaning of Hoare triples and the inference rules need adjustment. Similarly to \cite{GMB19}, we propose encoding a Hoare triple $\{b\}p\{c\}$ in \texttt{PKAT} as $b \cdot p \preccurlyeq b \cdot p \cdot c$, conveying that program correctness can only improve with execution. 



The study of the languages underlying paraconsistent transition structures is a line of work to be explored soon. In analogy to what is done in classic automata theory, we will consider a notion of ``paraconsistent automata'', by enriching the PLTS structure with a set of accepting states, in order to generate and characterize the algebra of the recognized expressions, say the paraconsistent regular languages. At this level, we expect to establish a Kleene-like Theorem and to frame such languages algebras as a PKAT.

Additionally, it is necessary to conduct a more thorough investigation into the potential applications and limitations that may arise from the flexibility of the adopted approach.


\bibliographystyle{eptcs}
\bibliography{mybibliography}
\end{document}